\documentclass[lettersize,journal]{IEEEtran}
\usepackage{amsmath,amsfonts}
\usepackage{amssymb}
\usepackage{algorithmic}
\usepackage{algorithm}
\usepackage{array}
\usepackage[caption=false,font=normalsize,labelfont=sf,textfont=sf]{subfig}
\usepackage{textcomp}
\usepackage{stfloats}
\usepackage{url}
\usepackage{verbatim}
\usepackage{graphicx}
\usepackage{amsthm}

\newtheorem{lemma}{Lemma}
\newtheorem{theorem}{Theorem}
\newtheorem{corollary}{Corollary}
\newtheorem{proposition}{Proposition}
\newtheorem{remark}{Remark}

\usepackage{cite}
\begin{document}

\title{Conflict-Aware Robust Design for Covert Wireless Communications}

\author{Abbas Arghavani
\thanks{A. Arghavani is with the Department of Computer Science and Engineering, M{\"a}lardalen University, V{\"a}ster{\aa}s, Sweden (email: abbas.arghavani@mdu.se)}
%\thanks{Manuscript received April 19, 2021; revised August 16, 2021.}
}

% The paper headers
\markboth{Journal of \LaTeX\ Class Files,~Vol.~X, No.~X, mm~202X}{A. Arghavani: Conflict-Aware Robust Covert Wireless Communication Under Bounded Uncertainty}%
%{Shell \MakeLowercase{\textit{et al.}}: Title}

%\IEEEpubid{0000--0000/00\$00.00~\copyright~202X IEEE}

\maketitle

\begin{abstract}
Covert wireless communication aims to establish a reliable link while hiding the existence of the transmission from an adversary. In wireless settings, uncertainty plays a central role in this tradeoff: it can help mask the signal from a warden, but it also complicates robust system design. This raises a basic question: under bounded uncertainty, are reliability and covertness governed by the same adverse conditions? If not, then robust covert design cannot be reduced to a single generic worst-case environment. In this paper, we study this question in a covert wireless model with bounded uncertainty, quasi-static fading, outage-based reliability at Bob, and radiometric detection at Willie. Uncertainty is represented through bounded intervals for Bob's average channel strength and Willie's noise power. To obtain a tractable analytical characterization, we adopt a conditional large-$N$ midpoint-threshold surrogate for Willie's detector, parameterized by a fixed Willie-side fading realization. Within this framework, we show that the reliability constraint is governed by Bob's smallest admissible channel parameter, whereas the covertness constraint is governed by Willie's smallest admissible noise level. This establishes a conflict-aware robust-design principle: the adverse realizations for reliability and covertness are different. Based on this structural result, we derive closed-form expressions for the robustly feasible transmit power and the corresponding robust optimal rate. Numerical results show that bounded uncertainty contracts the feasible region, monotonically reduces the robust optimal rate, and can cause substantial loss relative to the nominal design. Monte Carlo results further show that the conditional surrogate closely tracks the midpoint-threshold radiometer in the intended low-effective-SNR regime. Overall, the paper shows that even in a streamlined wireless setting, robust covert design requires different adverse-case reasoning for reliability and covertness.
\end{abstract}

\begin{center}
\footnotesize
\textit{This work has been submitted to the IEEE for possible publication.
Copyright may be transferred without notice, after which this version may no longer be accessible.}
\end{center}
\section{Introduction}
\label{sec:introduction}

\IEEEPARstart{C}{overt} communication concerns a stronger objective than ordinary secure communication. In classical secrecy, the main goal is to protect the message content from an eavesdropper. In covert communication, the goal is more demanding: the transmission itself should be difficult to detect \cite{bloch2016covert,wang2016fundamental,bash2015hiding,yan2019low}. In other words, even if an adversary listens to the channel, it should not be able to decide with high confidence whether communication is taking place.

This problem is especially relevant in wireless systems. Because wireless transmissions propagate through an open medium, they can often be observed by unintended receivers. A transmission may therefore reveal sensitive operational information even before any message is decoded. For this reason, covert communication has become an important topic in physical-layer security, especially for scenarios in which the mere presence of a transmission can be harmful \cite{bash2015hiding,yan2019low,chen2023survey}. This motivation also appears in emerging application domains such as in-body communication, where signal leakage may reveal sensitive medical information or even the presence of an implant itself \cite{padmal2025fat}.

A standard covert communication model involves three parties: a transmitter Alice, an intended receiver Bob, and a warden Willie. Alice wishes to send information to Bob, while Willie tries to determine whether Alice is transmitting or staying silent. This creates a fundamental tension. On the one hand, Alice must use enough signal power for Bob to decode reliably. On the other hand, increasing the signal power generally makes detection easier for Willie. Covert system design is therefore governed by a basic reliability-covertness tradeoff.

In wireless environments, this tradeoff is strongly influenced by uncertainty. The channel conditions observed by Bob and Willie are not perfectly stable, and the system designer may know them only approximately. Existing work has shown that uncertainty at the warden can play an important role in covert communication. Early studies demonstrated that uncertainty in the warden's noise power can significantly degrade detection performance and can even affect the fundamental scaling behavior of covert transmission \cite{goeckel2015covert,he2017covert}. Later works extended this line of research to settings involving interference uncertainty, fading, relay-assisted systems, and joint uncertainty in multiple parameters \cite{liu2018covert,shahzad2017covert,shahzad2020covert,wang2018covert,ta2019covert,ta2022covert}. Other studies considered more elaborate wireless architectures, including artificial noise, jamming, full-duplex receivers, relays, imperfect channel knowledge, and robust transceiver optimization \cite{shahzad2018achieving,soltani2018covert,li2020optimal,forouzesh2021robust,peng2021strategies,ma2021robust,he2022adaptive,he2023channel,jia2025robust,yang2024uav}.

At the same time, a covert system is only useful if Bob can still decode the message with acceptable reliability. In quasi-static fading channels, this requirement is naturally captured through outage probability, which has already been adopted in several covert wireless studies \cite{shahzad2017covert,shahzad2020covert,ta2019covert,arghavani2023covert}. As a result, practical covert wireless design must satisfy two requirements simultaneously: it must keep Bob's decoding failure sufficiently small while also keeping Willie's detection capability sufficiently weak.

Although the literature on covert wireless communication is now broad, one structural issue remains surprisingly unclear. When we design a covert wireless system robustly under bounded uncertainty, do the reliability and covertness constraints become most difficult under the same adverse parameter realization? This question matters because many robust formulations are intuitively interpreted through a single worst-case scenario. However, if the two constraints are stressed by different realizations, then that simplification is no longer valid. In that case, the design problem becomes inherently \emph{conflict-aware}: the legitimate link and the warden link are each governed by their own adverse conditions, and a feasible design must satisfy both at once.

In this paper, we study this structural question in a deliberately streamlined and analytically transparent wireless model. We consider one transmitter (Alice), one legitimate receiver (Bob), and one warden (Willie) under quasi-static fading. Reliability at Bob is measured by outage probability. Covertness at Willie is measured by the sum of false alarm and missed detection probabilities under radiometric detection. To isolate the effect of uncertainty without introducing unnecessary architectural complexity, we represent uncertainty through bounded intervals for two effective parameters: Bob's average channel strength and Willie's noise power. This focused model is amenable to closed-form analysis, yet still rich enough to reveal the robust-design structure of interest.

Our objective is not only to find a feasible transmit power under uncertainty, but also to understand how the robust covert design problem is organized. To do so, we first analyze a conditional large-$N$ midpoint-threshold surrogate for Willie's detector, parameterized by a fixed realization of the Willie-side fading gain. We then complement this with an averaged Rayleigh-fading benchmark on the Willie side. The analysis shows a clear asymmetry: the robust reliability constraint is governed by the smallest admissible Bob-side channel parameter, whereas the robust covertness constraint is governed by the smallest admissible Willie-side noise power. Therefore, even in this setting, robust covert wireless design is not driven by one common worst-case point. Instead, the two constraints are controlled by different adverse realizations.

The main contributions of this paper are summarized as follows:
\begin{enumerate}
\item We formulate a minimal robust covert wireless design problem under bounded uncertainty, quasi-static fading, outage-based reliability, and radiometric detection.

\item We derive a tractable Willie-side covertness surrogate based on a conditional large-$N$ midpoint-threshold approximation and use it to establish a conflict-aware robust-design principle. We further show, through an averaged Rayleigh-fading benchmark, that the same Willie-side worst-case noise boundary persists after averaging over the Willie-side fading gain.

\item We derive closed-form expressions for the robustly feasible transmit power and the corresponding robust optimal rate within this analytical framework.

\item We provide numerical results showing contraction of the robust feasible region, monotonic degradation of the robust optimal rate with uncertainty, and substantial rate loss relative to a nominal design.

\item We validate the analytical surrogate through Monte Carlo simulation of the radiometer under both the midpoint threshold and the exact likelihood-ratio threshold, and show that the two detector implementations are practically indistinguishable in the intended low-effective-SNR regime.
\end{enumerate}

The analysis is intentionally restricted to a focused setting so that the core structural message is easy to identify. In particular, the main closed-form design formulas are obtained under a conditional large-$N$ midpoint-threshold surrogate and are parameterized by a fixed Willie-side fading realization. This keeps the analysis tractable while still exposing the robust-design asymmetry of interest. The paper also includes an averaged Rayleigh-fading benchmark on the Willie side, and this benchmark leads to the same asymmetric-extremizer conclusion. Extensions to finite-block detection and richer adversarial uncertainty models are left for future work.

The remainder of the paper is organized as follows. Section~\ref{sec:related_work} reviews the most relevant literature. Sections~\ref{sec:system_model_minimal}--\ref{sec:main_results} develop the system model, robust formulation, and main analytical results. Section~\ref{sec:numerical_results} presents the numerical study. Section~\ref{sec:conclusion} concludes the paper.

\section{Related Work}
\label{sec:related_work}

Covert communication has its roots in foundational studies that established the basic meaning and limits of communication with a low probability of detection \cite{goeckel2015covert,bloch2016covert,wang2016fundamental,bash2015hiding}. These works clarified a key difference between covert communication and classical secrecy: in covert communication, the goal is not only to protect the message content, but also to make the transmission itself difficult for a warden to detect. They also showed that the warden's uncertainty can play a fundamental role in determining whether covert communication is possible and at what rate. Since then, a broad literature has developed on covert communication over noisy, fading, and state-dependent channels \cite{bloch2016covert,salehkalaibar2019covert,ahmadipour2019covert}. Broader background on the area is now available in survey and tutorial articles such as \cite{yan2019low,chen2023survey}. More recently, covert communication ideas have also begun to appear in emerging application domains beyond conventional over-the-air settings, including in-body communication, where tissue attenuation can itself contribute to covertness \cite{padmal2025fat}.

Within wireless covert communication, one major research direction studies how uncertainty at Willie can improve covertness. Early and influential works focused on uncertainty in the warden's noise power \cite{goeckel2015covert,he2017covert}. Later studies extended this idea to other sources of uncertainty, including interference, fading, channel variations, relay-assisted settings, and joint uncertainty in both channel and noise parameters \cite{liu2018covert,shahzad2017covert,shahzad2020covert,wang2018covert,ta2019covert,ta2022covert}. Related extensions have also examined multi-antenna detection and specialized wireless scenarios such as vehicular systems \cite{chen2023noise_multiantenna,duan2023covert}. Taken together, these works show that uncertainty is not merely an imperfection in covert wireless systems. In many cases, it is one of the main factors that shapes detectability and achievable performance.

A second line of research considers covert transmission under imperfect channel knowledge and related robust-design formulations. This includes, for example, transmission strategies under imperfect channel state information (CSI), robust power allocation, robust beamforming, adaptive power control under partial CSI, and robust designs for more specialized settings such as jammer-assisted, UAV, and satellite systems \cite{peng2021strategies,forouzesh2021robust,ma2021robust,he2022adaptive,jia2025robust,yang2024uav,xu2021robust,he2023channel}. These studies are important from a practical engineering perspective because they move beyond idealized assumptions and recognize that the designer often has only incomplete or imperfect knowledge of the wireless environment.

Another important theme in the literature is the tradeoff between covertness and useful communication performance. In fading channels, Bob-side reliability is often characterized through outage probability, effective covert rate, or related performance measures under covert constraints \cite{shahzad2017covert,shahzad2020covert,ta2019covert,arghavani2023covert}. There is also work that models covert transmission as a strategic interaction among transmitters, wardens, or assisting nodes \cite{arghavani2021game,arghavani2024dynamic}. More recent studies further consider stronger and more adaptive warden models, for example when Willie does not know the transmit power exactly and may exploit historical observations to improve detection \cite{he2024warden}. In parallel, a smaller body of work has begun to examine practical and application-specific covert communication settings with experimental components, including recent in-body communication results showing that tissue attenuation can support covert operation and that friendly jamming can further enhance covert throughput \cite{padmal2025fat}.

The present paper is related to these research directions, but its emphasis is different. Much of the existing literature introduces uncertainty either as a mechanism that can help covertness or as part of a robust design problem in a specific wireless architecture. Application-driven studies such as \cite{padmal2025fat} instead focus on the feasibility of covert operation in specialized physical environments. By contrast, our focus is on a more basic structural question: under bounded uncertainty, which adverse realizations govern the reliability constraint and which govern the covertness constraint? In the setting studied here, these two constraints are driven by different adverse parameters. This conflict-aware viewpoint is the main feature that distinguishes the present paper from prior work in wireless covert communication.

\section{System Model and Problem Formulation}
\label{sec:system_model_minimal}

We consider a covert wireless communication system with one transmitter Alice, one legitimate receiver Bob, and one warden Willie. Alice communicates with Bob over a wireless fading channel, while Willie observes the medium and attempts to decide whether Alice is transmitting or remaining silent. The wireless links are modeled as quasi-static fading over one transmission block, which is a standard and practically relevant model in prior covert wireless studies \cite{shahzad2020covert,ta2019covert}. Alice transmits over a block of $N$ channel uses.

\subsection{Signal Model}

Willie performs binary hypothesis testing between
\begin{align*}
\mathcal{H}_0 &: \text{Alice is silent},\\
\mathcal{H}_1 &: \text{Alice transmits}.
\end{align*}
Let $x[i]\sim\mathcal{CN}(0,1)$ denote the transmitted symbol, and let $P$ denote Alice's transmit power. Let $n_b[i]\sim\mathcal{CN}(0,\sigma_b^2)$ and $n_w[i]\sim\mathcal{CN}(0,\sigma_w^2)$ denote the receiver noises at Bob and Willie, respectively, where $\sigma_b^2$ and $\sigma_w^2$ are the corresponding noise powers. Let $h_b\sim\mathcal{CN}(0,\Omega_b)$ and $h_w\sim\mathcal{CN}(0,\Omega_w)$ denote the Alice--Bob and Alice--Willie fading coefficients, where $\Omega_b$ and $\Omega_w$ are the corresponding average channel powers.

Under these two hypotheses, the received signals at Bob and Willie are given by
\begin{align*}
y_b[i] &=
\begin{cases}
n_b[i], & \mathcal{H}_0,\\
h_b \sqrt{P}\, x[i] + n_b[i], & \mathcal{H}_1,
\end{cases}\\
y_w[i] &=
\begin{cases}
n_w[i], & \mathcal{H}_0,\\
h_w \sqrt{P}\, x[i] + n_w[i], & \mathcal{H}_1,
\end{cases}
\qquad i=1,\dots,N.
\end{align*}
We assume quasi-static block fading, so that $h_b$ and $h_w$ remain constant during one transmission block and change independently from one block to the next \cite{shahzad2020covert}.

In this model, reliability at Bob is measured through outage probability, while covertness at Willie is measured through the sum of false alarm and missed detection probabilities \cite{arghavani2023covert,shahzad2020covert,ta2019covert}.

\subsection{Uncertainty Model}

To isolate the robust-design effect of interest, we introduce bounded uncertainty directly into the two effective parameters that govern reliability and covertness:
\begin{align*}
\Omega_b &\in [\Omega_b^{-},\Omega_b^{+}],\\
\sigma_w^2 &\in [\sigma_w^{2,-},\sigma_w^{2,+}].
\end{align*}
Thus, uncertainty on the Bob side is represented through the average channel strength $\Omega_b$, while uncertainty on the Willie side is represented through the noise power $\sigma_w^2$. In the present model, $\Omega_w$ and $\sigma_b^2$ are treated as known constants.

The uncertainty in $\sigma_w^2$ should be interpreted from the designer's perspective. Specifically, Alice and Bob know only that Willie's true noise power lies in the interval $[\sigma_w^{2,-},\sigma_w^{2,+}]$. Willie, however, is conservatively assumed to know the realized value of $\sigma_w^2$ when configuring his detector. This avoids attributing an artificial disadvantage to the warden.

\subsection{Reliability Metric at Bob}

Alice transmits at a fixed coding rate $R$ bits/channel use. Following the outage-based treatment commonly used in quasi-static covert wireless models, Bob is said to be in outage whenever the instantaneous channel cannot support the target rate \cite{shahzad2020covert,arghavani2023covert}. The instantaneous signal-to-noise ratio (SNR) at Bob is
\begin{equation*}
\gamma_b = \frac{P|h_b|^2}{\sigma_b^2}.
\end{equation*}
Accordingly, the outage probability is
\begin{equation*}
P_{\mathrm{out}}(\Omega_b;P,R)
=
\Pr\!\left(
\log_2(1+\gamma_b) < R
\right).
\end{equation*}
Since $|h_b|^2$ is exponentially distributed with mean $\Omega_b$, this becomes
\begin{equation}
P_{\mathrm{out}}(\Omega_b;P,R)
=
1-\exp\!\left(
-\frac{(2^R-1)\sigma_b^2}{P\Omega_b}
\right).
\label{eq:pout_minimal}
\end{equation}
For fixed $(P,R)$, the outage probability decreases monotonically with $\Omega_b$. Therefore, the most adverse Bob-side realization for reliability is $\Omega_b=\Omega_b^{-}$.

\subsection{Covertness Metric at Willie}

Willie is modeled as a radiometer, or energy detector, with test statistic
\begin{equation}
T(y_w)=\frac{1}{N}\sum_{i=1}^N |y_w[i]|^2
\mathop{\gtrless}_{\mathcal H_0}^{\mathcal H_1} \lambda,
\label{eq:radiometer_clean}
\end{equation}
where $\lambda$ is the detection threshold. Based on this test, the false alarm probability $P_{\mathrm{FA}}$ is the probability that Willie decides Alice is transmitting when she is actually silent, while the missed detection probability $P_{\mathrm{MD}}$ is the probability that Willie decides Alice is silent when she is actually transmitting. These quantities are given by
\begin{align*}
P_{\mathrm{FA}}(\lambda;\sigma_w^2)
&= \Pr(T(y_w) \ge \lambda \mid \mathcal{H}_0),\\
P_{\mathrm{MD}}(\lambda;P,\sigma_w^2,g_w)
&= \Pr(T(y_w) < \lambda \mid \mathcal{H}_1, g_w).
\end{align*}

In the analytical framework of this paper, the Willie-side covertness metric is studied conditionally on the realized Willie-side fading power $g_w = |h_w|^2$, which is exponentially distributed with mean $\Omega_w$. Specifically, we define the minimum conditional total detection error probability as
\begin{equation}
\label{eq:xi_conditional_exact_definition}
\xi^\star_{\mathrm{cond}}(P,\sigma_w^2 \mid g_w)
\triangleq
\min_{\lambda}
\left[
P_{\mathrm{FA}}(\lambda;\sigma_w^2)
+
P_{\mathrm{MD}}(\lambda;P,\sigma_w^2,g_w)
\right].
\end{equation}
The corresponding conditional covertness requirement is
\begin{equation*}
\xi^\star_{\mathrm{cond}}(P,\sigma_w^2 \mid g_w)\ge 1-\epsilon,
\end{equation*}
where $\epsilon\in(0,1)$ is the allowed covertness deficit.

This conditional formulation is used as an analytical device to keep the Willie-side problem tractable. It should not be interpreted as assuming that Alice knows the realized Willie-side fading gain when selecting $(P,R)$. Rather, the conditioning is introduced only to expose the structure of the covertness constraint in a form that can later be reduced analytically.

\subsection{Robust Design Problem}

Alice chooses $(P,R)$ before the exact uncertain parameters are known. The Bob-side reliability requirement must hold for all $\Omega_b\in[\Omega_b^{-},\Omega_b^{+}]$, and the Willie-side covertness requirement must hold for all $\sigma_w^2\in[\sigma_w^{2,-},\sigma_w^{2,+}]$ within the conditional formulation above. Let $\delta\in(0,1)$ denote the maximum allowable outage probability at Bob, and let $P_{\max}$ denote the peak transmit-power budget. The resulting design problem is
\begin{align}
\max_{P,R}\quad & R
\label{eq:robust_design_problem}\\
\text{s.t.}\quad
& \sup_{\Omega_b\in[\Omega_b^{-},\Omega_b^{+}]}
P_{\mathrm{out}}(\Omega_b;P,R)\le \delta,\nonumber
\\
& \inf_{\sigma_w^2\in[\sigma_w^{2,-},\sigma_w^{2,+}]}
\xi^\star_{\mathrm{cond}}(P,\sigma_w^2 \mid g_w)\ge 1-\epsilon,
\label{eq:robust_covertness_constraint}\\
& 0\le P\le P_{\max}.\nonumber
\end{align}

From \eqref{eq:pout_minimal}, the reliability constraint is already seen to be governed by the smallest admissible value $\Omega_b^{-}$. By contrast, the covertness constraint is still expressed through the exact conditional quantity $\xi^\star_{\mathrm{cond}}(P,\sigma_w^2 \mid g_w)$. The next section develops a tractable large-$N$ surrogate for this Willie-side term and then uses it to obtain a closed-form robust reduction. For comparison, the numerical results also include a nominal benchmark obtained by evaluating the same analytical framework at the reference values $\Omega_{b,0}$ and $\sigma_{w,0}^2$, rather than at the adverse realizations.

\section{Willie Detection Analysis and Robust Power Limit}
\label{sec:willie_analysis}

We now derive a tractable approximation for the Willie-side covertness metric in \eqref{eq:robust_covertness_constraint}. The goal is to replace the exact conditional covertness term by a simple analytical surrogate that yields an explicit robust power constraint.

\subsection{Conditional Detection Model and Large-$N$ Surrogate}

Define
\begin{equation*}
g_w \triangleq |h_w|^2.
\end{equation*}
Conditioned on the realized value of $g_w$, and using the Gaussian signaling assumption $x[i]\sim\mathcal{CN}(0,1)$, we have
\begin{align*}
y_w[i]\mid \mathcal H_0 &\sim \mathcal{CN}(0,\sigma_w^2),\\
y_w[i]\mid \mathcal H_1,g_w &\sim \mathcal{CN}(0,\sigma_w^2 + Pg_w).
\end{align*}

Since the squared magnitude of a circularly symmetric complex Gaussian random variable is exponentially distributed, the received signal power $|y_w[i]|^2$ under each hypothesis follows an exponential distribution. Consequently, the radiometer statistic in \eqref{eq:radiometer_clean} satisfies
\begin{align*}
\mathbb{E}[T \mid \mathcal H_0] &= \sigma_w^2,\\
\mathrm{Var}(T \mid \mathcal H_0) &= \frac{\sigma_w^4}{N},\\
\mathbb{E}[T \mid \mathcal H_1,g_w] &= \sigma_w^2 + Pg_w,\\
\mathrm{Var}(T \mid \mathcal H_1,g_w) &= \frac{(\sigma_w^2+Pg_w)^2}{N}.
\end{align*}
For large $N$, the central limit theorem gives
\begin{align*}
T\mid \mathcal H_0
&\approx
\mathcal N\!\left(
\sigma_w^2,\,
\frac{\sigma_w^4}{N}
\right),
\\
T\mid \mathcal H_1,g_w
&\approx
\mathcal N\!\left(
\sigma_w^2+Pg_w,\,
\frac{(\sigma_w^2+Pg_w)^2}{N}
\right).
\end{align*}

Starting from the exact conditional definition in \eqref{eq:xi_conditional_exact_definition}, we focus on the low-effective-SNR regime
\begin{equation}
Pg_w \ll \sigma_w^2,
\label{eq:low_snr_clean}
\end{equation}
in which Willie's total detection error remains close to one. Under \eqref{eq:low_snr_clean}, the two hypothesis-dependent variances are approximately equal:
\begin{equation*}
\frac{(\sigma_w^2+Pg_w)^2}{N}
\approx
\frac{\sigma_w^4}{N}.
\end{equation*}
Hence, we use a common variance
\begin{equation*}
v_T \triangleq \frac{\sigma_w^4}{N},
\end{equation*}
and means
\begin{equation*}
\mu_0 \triangleq \sigma_w^2,
\qquad
\mu_1 \triangleq \sigma_w^2+Pg_w.
\end{equation*}
For equal-variance Gaussian hypotheses, the threshold minimizing $P_{\mathrm{FA}}+P_{\mathrm{MD}}$ is the midpoint
\begin{equation}
\lambda^\star
=
\frac{\mu_0+\mu_1}{2}
=
\sigma_w^2+\frac{Pg_w}{2}.
\label{eq:midpoint_threshold}
\end{equation}
Substituting \eqref{eq:midpoint_threshold} into the Gaussian approximation yields the surrogate
\begin{equation}
\xi_{\mathrm{cond}}^{(\mathrm{apx})}(P,\sigma_w^2 \mid g_w)
\triangleq
2Q\!\left(\frac{\sqrt{N}\,P\,g_w}{2\sigma_w^2}\right).
\label{eq:xi_conditional_approx}
\end{equation}
Equation \eqref{eq:xi_conditional_approx} is the key working approximation used in the remainder of the analysis, where $Q(\cdot)$ denotes the standard Gaussian tail probability function.

\subsection{Monotonicity and Robust Covertness Reduction}

The surrogate in \eqref{eq:xi_conditional_approx} has two simple monotonicity properties that drive the robust reduction.

\begin{lemma}[Monotonicity in Willie-side noise power]
\label{lem:xi_sigma}
For fixed $P>0$, $N>0$, and $g_w>0$, the surrogate
\[
\sigma_w^2 \mapsto \xi_{\mathrm{cond}}^{(\mathrm{apx})}(P,\sigma_w^2 \mid g_w)
\]
is strictly increasing.
\end{lemma}

\begin{proof}
From \eqref{eq:xi_conditional_approx},
\[
\xi_{\mathrm{cond}}^{(\mathrm{apx})}(P,\sigma_w^2 \mid g_w)
=
2Q\!\left(\frac{\sqrt{N}\,P\,g_w}{2\sigma_w^2}\right).
\]
For fixed $P>0$, let
\[
a \triangleq \frac{\sqrt{N}\,P\,g_w}{2} > 0.
\]
Then
\[
\xi_{\mathrm{cond}}^{(\mathrm{apx})}(P,\sigma_w^2 \mid g_w)
=
2Q\!\left(\frac{a}{\sigma_w^2}\right).
\]
Since $Q(\cdot)$ is strictly decreasing and $a/\sigma_w^2$ is strictly decreasing in $\sigma_w^2$, the result follows.
\end{proof}

\begin{lemma}[Monotonicity in transmit power]
\label{lem:xi_power}
For fixed $\sigma_w^2>0$, $N>0$, and $g_w>0$, the surrogate
\[
P \mapsto \xi_{\mathrm{cond}}^{(\mathrm{apx})}(P,\sigma_w^2 \mid g_w)
\]
is strictly decreasing.
\end{lemma}

\begin{proof}
From \eqref{eq:xi_conditional_approx},
\[
\xi_{\mathrm{cond}}^{(\mathrm{apx})}(P,\sigma_w^2 \mid g_w)
=
2Q\!\left(\frac{\sqrt{N}\,P\,g_w}{2\sigma_w^2}\right).
\]
For fixed $\sigma_w^2>0$, the argument of the $Q$-function is strictly increasing in $P$. Since $Q(\cdot)$ is strictly decreasing, the result follows.
\end{proof}

By Lemma~\ref{lem:xi_sigma},
\begin{equation}
\inf_{\sigma_w^2 \in [\sigma_w^{2,-},\,\sigma_w^{2,+}]}
\xi_{\mathrm{cond}}^{(\mathrm{apx})}(P,\sigma_w^2 \mid g_w)
=
\xi_{\mathrm{cond}}^{(\mathrm{apx})}(P,\sigma_w^{2,-} \mid g_w).
\label{eq:conditional_covertness_reduction}
\end{equation}
Hence, the robust covertness constraint reduces to
\begin{equation}
2Q\!\left(
\frac{\sqrt{N}\,P\,g_w}{2\sigma_w^{2,-}}
\right)
\ge 1-\epsilon.
\label{eq:conditional_covertness_constraint_reduced}
\end{equation}
Solving \eqref{eq:conditional_covertness_constraint_reduced} for $P$ gives the covertness-induced power ceiling
\begin{equation*}
P \le P^\star_{\mathrm{cov}}
\triangleq
\frac{2\sigma_w^{2,-}}{\sqrt{N}\,g_w}
Q^{-1}\!\left(\frac{1-\epsilon}{2}\right).
\end{equation*}
Therefore, the feasible transmit power must satisfy
\begin{equation*}
0 \le P \le \min\{P_{\max},\,P^\star_{\mathrm{cov}}\}.
\end{equation*}

\begin{proposition}[Averaged Willie-side benchmark]
\label{prop:avg_willie_benchmark}
Assume
\[
g_w \sim \mathrm{Exp}\!\left(\frac{1}{\Omega_w}\right),
\]
and define the averaged surrogate
\[
\bar{\xi}_{\mathrm{cond}}^{(\mathrm{apx})}(P,\sigma_w^2)
\triangleq
\mathbb{E}_{g_w}
\!\left[
\xi_{\mathrm{cond}}^{(\mathrm{apx})}(P,\sigma_w^2 \mid g_w)
\right].
\]
Then, for fixed $P>0$,
\begin{equation}
\begin{aligned}
&\bar{\xi}_{\mathrm{cond}}^{(\mathrm{apx})}(P,\sigma_w^2)
=
1
-
2\exp\!\left(\frac{\eta^2}{2}\right)Q(\eta),
\\
&\eta \triangleq \frac{2\sigma_w^2}{\sqrt{N}\,P\,\Omega_w}.
\label{eq:xi_avg_closed_form}
\end{aligned}
\end{equation}
Moreover, for fixed $P>0$, $N>0$, and $\Omega_w>0$, the function
\[
\sigma_w^2 \mapsto \bar{\xi}_{\mathrm{cond}}^{(\mathrm{apx})}(P,\sigma_w^2)
\]
is strictly increasing. Consequently,
\begin{equation}
\inf_{\sigma_w^2 \in [\sigma_w^{2,-},\,\sigma_w^{2,+}]}
\bar{\xi}_{\mathrm{cond}}^{(\mathrm{apx})}(P,\sigma_w^2)
=
\bar{\xi}_{\mathrm{cond}}^{(\mathrm{apx})}(P,\sigma_w^{2,-}).
\label{eq:avg_covertness_reduction}
\end{equation}
\end{proposition}

\begin{proof}
Let
\[
a \triangleq \frac{\sqrt{N}\,P}{2\sigma_w^2},
\qquad
f_{g_w}(g)=\frac{1}{\Omega_w}e^{-g/\Omega_w}, \quad g\ge 0.
\]
Using \eqref{eq:xi_conditional_approx},
\[
\bar{\xi}_{\mathrm{cond}}^{(\mathrm{apx})}(P,\sigma_w^2)
=
\frac{2}{\Omega_w}
\int_0^\infty
Q(ag)\,e^{-g/\Omega_w}\,dg.
\]
Now write
\[
Q(x)=\int_x^\infty \phi(t)\,dt,
\qquad
\phi(t)=\frac{1}{\sqrt{2\pi}}e^{-t^2/2}.
\]
Exchanging the order of integration gives
\[
\bar{\xi}_{\mathrm{cond}}^{(\mathrm{apx})}(P,\sigma_w^2)
=
2\int_0^\infty
\phi(t)
\left(
1-e^{-t/(a\Omega_w)}
\right)
dt.
\]
Hence,
\[
\bar{\xi}_{\mathrm{cond}}^{(\mathrm{apx})}(P,\sigma_w^2)
=
1
-
2\int_0^\infty
\phi(t)e^{-\beta t}\,dt,
\quad
\beta \triangleq \frac{1}{a\Omega_w}.
\]
Completing the square yields
\[
\int_0^\infty \phi(t)e^{-\beta t}\,dt
=
e^{\beta^2/2}Q(\beta),
\]
so that
\[
\bar{\xi}_{\mathrm{cond}}^{(\mathrm{apx})}(P,\sigma_w^2)
=
1-2e^{\beta^2/2}Q(\beta).
\]
Substituting
\[
\beta=\frac{2\sigma_w^2}{\sqrt{N}\,P\,\Omega_w}
\]
gives \eqref{eq:xi_avg_closed_form}. Strict monotonicity in $\sigma_w^2$ follows because the integrand in \eqref{eq:xi_conditional_approx} is strictly increasing in $\sigma_w^2$ for every $g_w>0$ by Lemma~\ref{lem:xi_sigma}, and the exponential distribution has full support on $(0,\infty)$. The reduction \eqref{eq:avg_covertness_reduction} follows immediately.
\end{proof}

The averaged benchmark in Proposition~\ref{prop:avg_willie_benchmark} does not yield the same simple closed-form power ceiling as the conditional design formulas developed above, but it confirms that the same Willie-side adverse noise boundary $\sigma_w^{2,-}$ remains governing even after averaging over Rayleigh fading.

\section{Main Structural Results}
\label{sec:main_results}

We now combine the Bob-side monotonicity of the outage probability with the Willie-side monotonicity of the surrogate covertness metric to obtain the main structural conclusion of the paper.

\begin{lemma}[Worst-case Bob-side realization]
\label{lem:bob_worst_case}
For fixed $P>0$ and $R>0$, the outage probability
$P_{\mathrm{out}}(\Omega_b;P,R)$ in \eqref{eq:pout_minimal} is strictly decreasing in $\Omega_b$. Consequently,
\[
\sup_{\Omega_b \in [\Omega_b^{-},\,\Omega_b^{+}]}
P_{\mathrm{out}}(\Omega_b;P,R)
=
P_{\mathrm{out}}(\Omega_b^{-};P,R).
\]
\end{lemma}

\begin{proof}
Let
\[
a(P,R) \triangleq \frac{(2^R-1)\sigma_b^2}{P} > 0.
\]
Using \eqref{eq:pout_minimal},
\[
P_{\mathrm{out}}(\Omega_b;P,R)
=
1-\exp\!\left(-\frac{a(P,R)}{\Omega_b}\right).
\]
Differentiating with respect to $\Omega_b$ gives
\[
\frac{\partial P_{\mathrm{out}}}{\partial \Omega_b}
=
-
\frac{a(P,R)}{\Omega_b^2}
\exp\!\left(-\frac{a(P,R)}{\Omega_b}\right)
<0,
\]
which proves strict monotonicity.
\end{proof}

\begin{theorem}[Asymmetric extremizers]
\label{thm:asymmetric_extremizers}
Consider the minimal model under quasi-static fading, outage-based reliability, and the surrogate in \eqref{eq:xi_conditional_approx}. Then
\begin{enumerate}
\item the robust reliability constraint is attained at $\Omega_b = \Omega_b^{-}$;
\item the robust covertness constraint is attained at $\sigma_w^2 = \sigma_w^{2,-}$;
\item therefore, robust reliability and robust covertness are governed by different adverse parameters on the legitimate and warden links, respectively.
\end{enumerate}
\end{theorem}

\begin{proof}
Item 1 follows from Lemma~\ref{lem:bob_worst_case}. Item 2 follows from Lemma~\ref{lem:xi_sigma}, since $\xi_{\mathrm{cond}}^{(\mathrm{apx})}(P,\sigma_w^2 \mid g_w)$ is strictly increasing in $\sigma_w^2$. Item 3 follows immediately from Items 1 and 2.
\end{proof}

\begin{corollary}[Reduced robust constraints]
\label{cor:reduced_robust_constraints}
Under Theorem~\ref{thm:asymmetric_extremizers}, the robust design problem in \eqref{eq:robust_design_problem} reduces to the equivalent form
\begin{align}
\max_{P,R}\quad & R
\nonumber\\
\text{s.t.}\quad
& 1-\exp\!\left(
-\frac{(2^R-1)\sigma_b^2}{P\Omega_b^-}
\right)\le \delta,
\label{eq:reduced_reliability_constraint}\\
& 2Q\!\left(
\frac{\sqrt{N}\,P\,g_w}{2\sigma_w^{2,-}}
\right)\ge 1-\epsilon,
\label{eq:reduced_covertness_constraint}\\
& 0 \le P \le P_{\max}.
\label{eq:reduced_power_constraint}
\end{align}
\end{corollary}

\begin{proof}
The reliability reduction follows from Lemma~\ref{lem:bob_worst_case}, and the covertness reduction follows from \eqref{eq:conditional_covertness_reduction}.
\end{proof}

\begin{remark}[Conditional and averaged viewpoints]
\label{rem:conditional_interpretation}
The closed-form power bound and robust-rate expression developed below are parameterized by a given realization of $g_w = |h_w|^2$, which keeps the main design formulas simple and transparent. At the same time, Proposition~\ref{prop:avg_willie_benchmark} shows that after averaging over Willie-side Rayleigh fading, the same Willie-side adverse noise boundary $\sigma_w^{2,-}$ remains governing. Hence, the asymmetric-extremizer conclusion is not an artifact of conditioning on a fixed $g_w$.
\end{remark}

\begin{proposition}[Conditional robust optimal power and rate]
\label{prop:robust_opt_power_rate}
Under the reduced robust constraints in Corollary~\ref{cor:reduced_robust_constraints}, the largest robustly feasible transmit power is
\begin{equation}
P^\star
=
\min\!\left\{
P_{\max},
\frac{2\sigma_w^{2,-}}{\sqrt{N}\,g_w}
Q^{-1}\!\left(\frac{1-\epsilon}{2}\right)
\right\}.
\label{eq:pstar_robust}
\end{equation}
The corresponding maximum robust coding rate is
\begin{equation}
R^\star_{\mathrm{rob}}
=
\log_2\!\left(
1+
\frac{P^\star \Omega_b^-}{\sigma_b^2}
\ln\!\frac{1}{1-\delta}
\right).
\label{eq:rstar_robust}
\end{equation}
\end{proposition}

\begin{proof}
From \eqref{eq:reduced_covertness_constraint},
\[
P \le
\frac{2\sigma_w^{2,-}}{\sqrt{N}\,g_w}
Q^{-1}\!\left(\frac{1-\epsilon}{2}\right).
\]
Combining this with \eqref{eq:reduced_power_constraint} gives \eqref{eq:pstar_robust}. Next, \eqref{eq:reduced_reliability_constraint} implies
\[
R \le
\log_2\!\left(
1+
\frac{P\Omega_b^-}{\sigma_b^2}
\ln\!\frac{1}{1-\delta}
\right).
\]
Since the right-hand side is strictly increasing in $P$, the maximal feasible rate is attained at $P=P^\star$, which yields \eqref{eq:rstar_robust}.
\end{proof}

Theorem~\ref{thm:asymmetric_extremizers} and Proposition~\ref{prop:robust_opt_power_rate} together convey the main structural message of the paper. Bob-side robustness is dictated by the smallest admissible Bob-side channel parameter, whereas Willie-side robustness is dictated by the smallest admissible Willie-side noise power. Robust covert design is therefore conflict-aware: the two constraints are governed by different adverse parameters.

\section{Numerical Results}
\label{sec:numerical_results}

This section illustrates the implications of the proposed robust design framework. The numerical study has three main objectives: to visualize the difference between nominal and robust feasible designs, to quantify the rate penalty caused by uncertainty, and to assess the accuracy of the conditional large-$N$ midpoint-threshold surrogate used in the analysis, including comparison with an exact-threshold radiometer benchmark.

The numerical results are aligned with the analytical model developed in Sections~\ref{sec:willie_analysis} and~\ref{sec:main_results}. Accordingly, unless otherwise stated, the Willie-side results are evaluated at a fixed representative fading realization $g_w$.

\subsection{Setup and Baseline Parameters}

Unless otherwise stated, the baseline parameters are listed in Table~\ref{tab:baseline_parameters}.
\begin{table}[!b]
\centering
\caption{Baseline parameters used in the numerical results unless otherwise stated.}
\label{tab:baseline_parameters}
\begin{tabular}{c c c c}
\hline
Parameter & Value & Parameter & Value \\
\hline
$\Omega_{b,0}$ & 1 & $\sigma_{w,0}^2$ & 1 \\
$\sigma_b^2$ & 1 & $P_{\max}$ & 10 \\
$N$ & 100 & $\delta$ & 0.1 \\
$\epsilon$ & 0.2 & $g_w$ & 0.2 \\
\hline
\end{tabular}
\end{table}
For the one-dimensional uncertainty sweeps, we use the symmetric parameterization
\begin{equation*}
\Omega_b^-=\Omega_{b,0}(1-u),
\qquad
\sigma_w^{2,-}=\sigma_{w,0}^2(1-u),
\end{equation*}
where $u\in[0,0.6]$ denotes the common uncertainty width. When analyzing the Bob-side and Willie-side uncertainty sources separately, we denote their widths by $u_b$ and $u_w$, respectively. Unless explicitly stated otherwise, the nominal benchmark uses the reference values $\Omega_{b,0}$ and $\sigma_{w,0}^2$.

\subsection{Feasible Region and Rate Degradation}

Fig.~\ref{fig:feasible_region} compares the nominal and robust feasible regions in the $(P,R)$ plane. The nominal design uses the reference parameters $(\Omega_{b,0},\sigma_{w,0}^2)$, whereas the robust design enforces the constraints at the adverse realizations $(\Omega_b^-,\sigma_w^{2,-})$. For this figure, we use $u_b=u_w=0.2$.

Two observations are immediate. First, the robust feasible region is strictly smaller than the nominal one. Second, the reduction occurs through two different mechanisms: the reliability ceiling is lowered by the weaker Bob-side parameter $\Omega_b^-$, whereas the feasible power ceiling is tightened by the smaller Willie-side noise level $\sigma_w^{2,-}$. The robust operating point therefore shifts both leftward and downward relative to the nominal operating point.

\begin{figure}[t]
    \centering
    \includegraphics[width=0.9\linewidth]{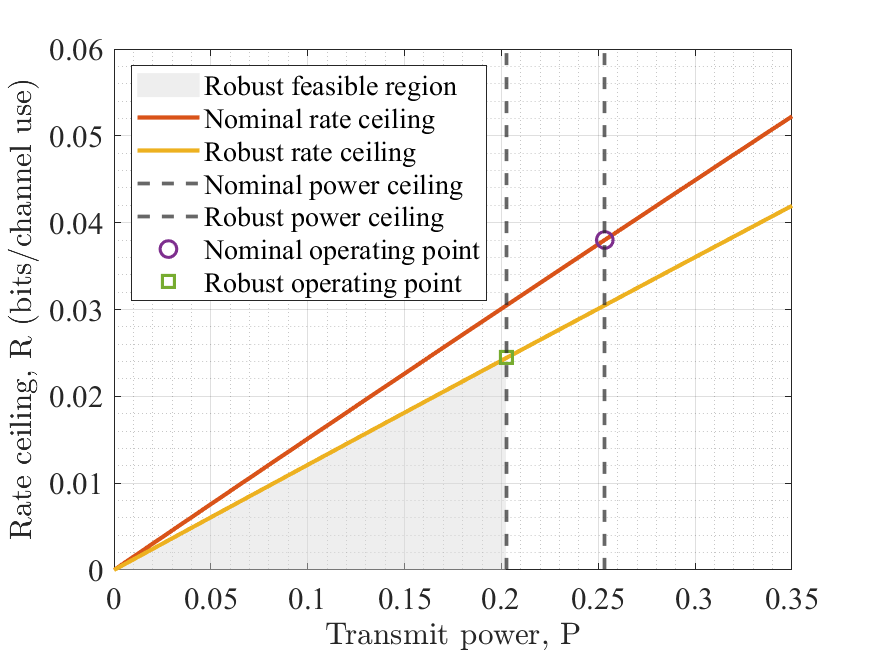}
    \caption{Nominal and robust feasible regions in the $(P,R)$ plane. The robust design yields both a lower rate ceiling and a tighter power ceiling, shifting the operating point downward and leftward relative to the nominal design.}
    \label{fig:feasible_region}
\end{figure}

Fig.~\ref{fig:rate_vs_uncertainty} shows the nominal and robust optimal rates as functions of the common uncertainty width $u$. As expected, the two designs coincide at $u=0$, since the robust and nominal formulations are then identical. As $u$ increases, however, the robust optimal rate decreases monotonically, while the nominal benchmark remains fixed because it continues to use the reference parameter values.

\begin{figure}[t]
    \centering
    \includegraphics[width=0.9\linewidth]{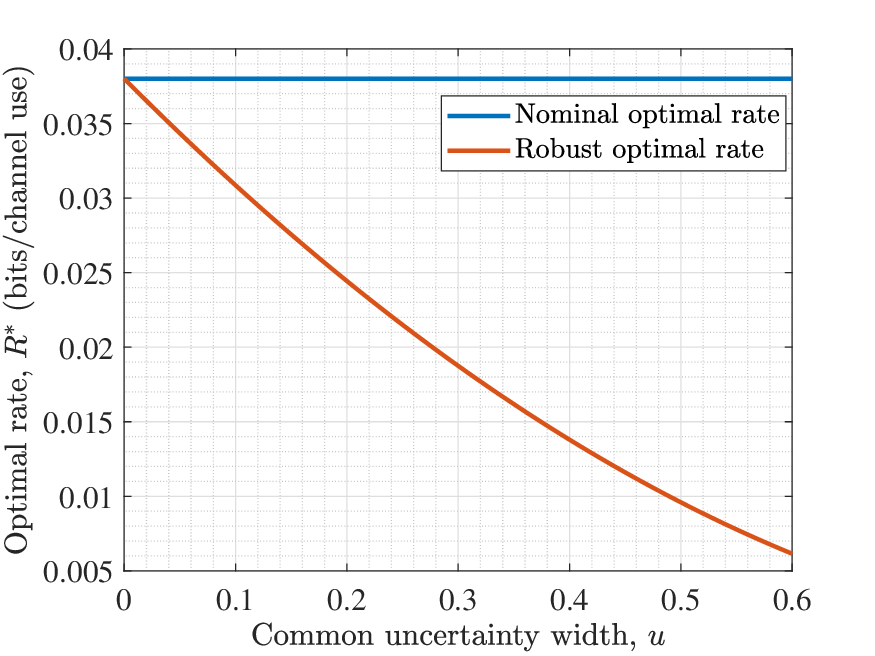}
    \caption{Nominal and robust optimal rates versus the common uncertainty width $u$. The robust rate decreases monotonically with uncertainty, while the nominal benchmark remains fixed in this setup.}
    \label{fig:rate_vs_uncertainty}
\end{figure}

To make the penalty more explicit, Fig.~\ref{fig:relative_loss} plots the relative rate loss
\begin{equation*}
\Delta_R
=
\frac{R_{\mathrm{nom}}^\star-R_{\mathrm{rob}}^\star}
{R_{\mathrm{nom}}^\star},
\end{equation*}
where $R_{\mathrm{nom}}^\star$ is the optimal rate of the nominal design. The loss starts at zero and grows monotonically with uncertainty. Even in this deliberately minimal model, bounded uncertainty produces a substantial covert-throughput penalty.

\begin{figure}[t]
    \centering
    \includegraphics[width=0.9\linewidth]{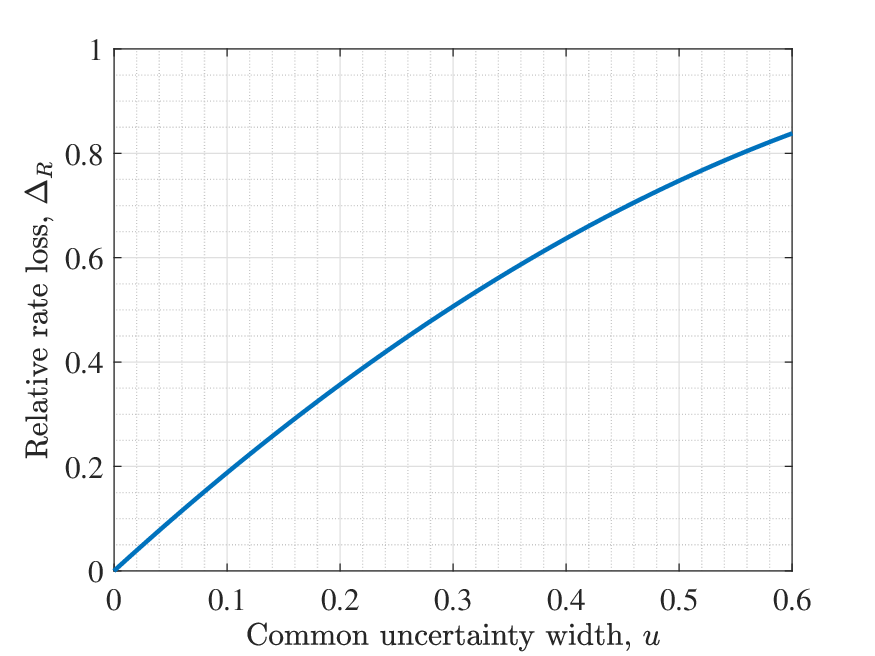}
    \caption{Relative robust-rate loss versus the common uncertainty width $u$. The loss grows monotonically, showing that uncertainty can impose a substantial covert-throughput penalty.}
    \label{fig:relative_loss}
\end{figure}

\subsection{Two-Dimensional Robustness Landscape and Power Ceiling}

Fig.~\ref{fig:heatmap_robust_rate} presents a two-dimensional heatmap of the robust optimal rate $R_{\mathrm{rob}}^\star(u_b,u_w)$. The heatmap is generated over independent grids $u_b,u_w\in[0,0.6]$. The largest rates occur near $(u_b,u_w)=(0,0)$, while the lowest rates occur when both uncertainty widths are large.

This figure shows that uncertainty on either side reduces the achievable robust covert rate, even though the two constraints are governed by different adverse parameters. Under the symmetric parameterization adopted here, the numerical values also exhibit the symmetry
\begin{equation*}
R_{\mathrm{rob}}^\star(u_b,0)=R_{\mathrm{rob}}^\star(0,u_b).
\end{equation*}
In the present model, this symmetry is induced by the simplified robust-rate expression together with the symmetric parameterization of the uncertainty widths. It therefore concerns the final rate values, not the asymmetric-extremizer result itself.

\begin{figure}[t]
    \centering
    \includegraphics[width=0.9\linewidth]{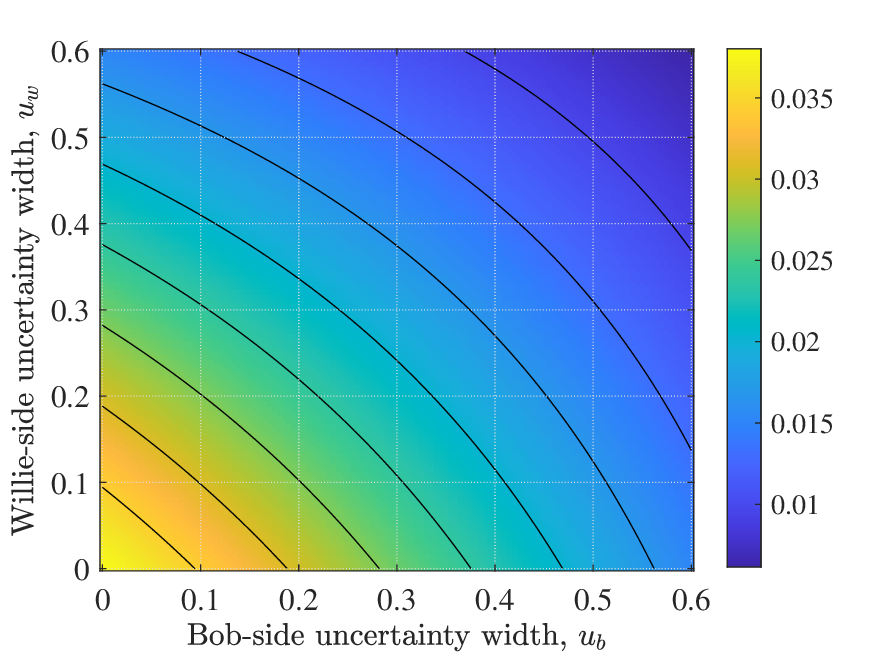}
    \caption{Heatmap of the robust optimal rate $R_{\mathrm{rob}}^\star(u_b,u_w)$ as a function of Bob-side and Willie-side uncertainty widths. The rate decreases as either uncertainty source increases.}
    \label{fig:heatmap_robust_rate}
\end{figure}

Fig.~\ref{fig:power_ceiling} helps explain the rate degradation by showing the feasible power ceiling as a function of uncertainty width. The nominal power ceiling remains constant, whereas the robust power ceiling decreases as uncertainty grows. This follows directly from the reduced covertness constraint, since the minimum admissible Willie-side noise power decreases with uncertainty. The reduced power budget then lowers the maximum achievable rate through the Bob-side outage constraint.

\begin{figure}[t]
    \centering
    \includegraphics[width=0.9\linewidth]{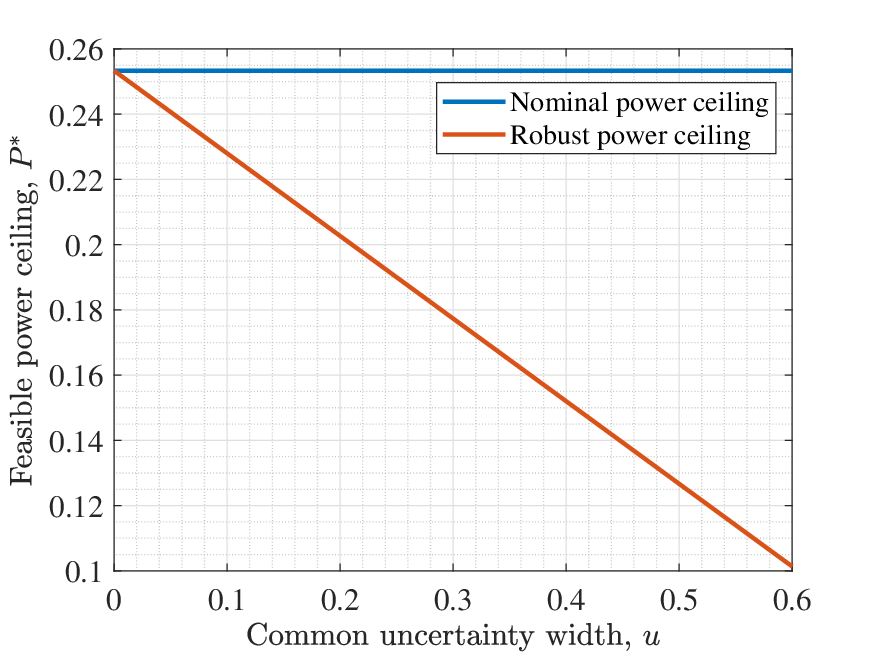}
    \caption{Nominal and robust feasible power ceilings versus the common uncertainty width $u$. The robust power ceiling decreases with uncertainty, explaining the reduction in robust optimal rate.}
    \label{fig:power_ceiling}
\end{figure}

Table~\ref{tab:representative_results} reports representative numerical values for several uncertainty settings. It confirms the same trends seen in the figures: both the feasible power and the robust optimal rate decrease with uncertainty, and the relative rate loss can become large even at moderate uncertainty widths.

\begin{table}[t]
\centering
\caption{Representative numerical results for the robust design under different uncertainty settings. Here $P_{\mathrm{nom}}^\star$ and $R_{\mathrm{nom}}^\star$ denote the optimal power and rate of the nominal design, with values $P_{\mathrm{nom}}^\star = 0.253347$ and $R_{\mathrm{nom}}^\star = 0.038005$. The relative rate loss is defined as $\Delta_R=(R_{\mathrm{nom}}^\star-R_{\mathrm{rob}}^\star)/R_{\mathrm{nom}}^\star$.}
\label{tab:representative_results}
\begin{tabular}{c c c c c}
\hline
$u_b$ & $u_w$ & $P^\star$ & $R_{\mathrm{rob}}^\star$ & $\Delta_R$ (\%) \\
\hline
0.0 & 0.0 & 0.253347 & 0.038005 & 0.00 \\
0.1 & 0.1 & 0.228012 & 0.030860 & 18.80 \\
0.2 & 0.2 & 0.202678 & 0.024438 & 35.70 \\
0.3 & 0.3 & 0.177343 & 0.018747 & 50.67 \\
0.4 & 0.4 & 0.152008 & 0.013797 & 63.70 \\
0.6 & 0.6 & 0.101339 & 0.006148 & 83.82 \\
0.3 & 0.0 & 0.253347 & 0.026708 & 29.72 \\
0.0 & 0.3 & 0.177343 & 0.026708 & 29.72 \\
0.6 & 0.0 & 0.253347 & 0.015322 & 59.68 \\
0.0 & 0.6 & 0.101339 & 0.015322 & 59.68 \\
\hline
\end{tabular}
\end{table}

\subsection{Validation of the Large-$N$ Approximation and Threshold Choice}

Fig.~\ref{fig:xi_validation} compares the analytical surrogate in \eqref{eq:xi_conditional_approx} with two Monte Carlo benchmarks: one using the midpoint threshold in \eqref{eq:midpoint_threshold}, and one using the exact threshold that minimizes $P_{\mathrm{FA}}+P_{\mathrm{MD}}$ for the radiometer statistic.

Under $\mathcal H_0$ and under $\mathcal H_1$ conditioned on the fixed value of $g_w$, the statistic $T(y_w)$ is gamma distributed with common shape parameter $N$ and different scale parameters, with means
\[
\mu_0=\sigma_w^2,
\qquad
\mu_1=\sigma_w^2+Pg_w.
\]
Since the likelihood ratio is monotone in $T(y_w)$, the exact threshold is obtained by equating the two gamma densities, which gives
\[
\lambda^\star_{\mathrm{ex}}
=
\frac{\mu_0\mu_1}{\mu_1-\mu_0}
\ln\!\left(\frac{\mu_1}{\mu_0}\right).
\]

The midpoint threshold and $\lambda^\star_{\mathrm{ex}}$ are nearly identical in the intended low-effective-SNR regime $Pg_w \ll \sigma_w^2$. Consequently, the two Monte Carlo curves are visually indistinguishable in Fig.~\ref{fig:xi_validation}. This confirms that, in the operating regime relevant to the paper, the midpoint-threshold approximation is not only analytically convenient but also practically accurate. At the same time, the analytical surrogate remains close to both Monte Carlo benchmarks over the tested power range, with only a modest deviation emerging as the transmit power increases and the system moves away from the equal-variance regime. Across the tested power range, the analytical surrogate exhibits a mean absolute error of about $4\times 10^{-3}$ and a maximum absolute error of about $1.2\times 10^{-2}$ relative to the Monte Carlo benchmarks.

\begin{figure}[t]
    \centering
    \includegraphics[width=0.9\linewidth]{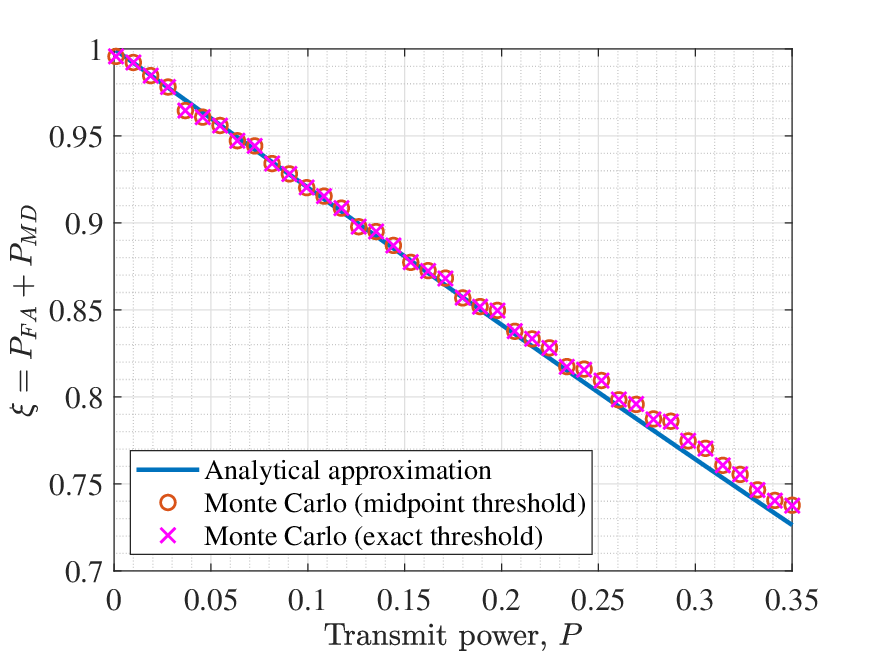}
    \caption{Validation of the conditional large-$N$ midpoint-threshold surrogate for Willie's detection error against Monte Carlo radiometer benchmarks using both the midpoint threshold and the exact likelihood-ratio threshold. In the intended low-effective-SNR regime, the two Monte Carlo curves are visually indistinguishable, confirming that the midpoint threshold is an accurate approximation in this setting.}
    \label{fig:xi_validation}
\end{figure}

Overall, the numerical results support the analytical message of the paper. The robust feasible set is smaller than the nominal one, the robust optimal rate decreases systematically with uncertainty, and the resulting performance loss can be substantial. In addition, the Willie-side validation shows that the midpoint-threshold and exact-threshold radiometer benchmarks are visually indistinguishable in the intended low-effective-SNR regime, while both remain close to the analytical surrogate over the tested power range.

\section{Conclusion}
\label{sec:conclusion}

This paper studied robust covert wireless communication under bounded uncertainty in a minimal single-warden setting. Using quasi-static fading, outage-based reliability at Bob, and radiometric detection at Willie, we formulated a robust design problem in which uncertainty enters through Bob's average channel strength and Willie's noise power. The main goal was to understand how bounded uncertainty shapes the structure of robust covert design.

The analysis showed that the worst-case reliability constraint is governed by the smallest admissible Bob-side channel parameter, whereas the worst-case covertness constraint is governed by the smallest admissible Willie-side noise level within the adopted analytical framework. The same Willie-side adverse noise boundary also remains governing under an averaged Rayleigh-fading benchmark on the warden link. Hence, even in this streamlined model, robust covert design is conflict-aware: the adverse realization for reliability differs from that for covertness. Based on this observation, we derived closed-form expressions for the robustly feasible transmit power and the corresponding robust optimal rate under a conditional large-$N$ midpoint-threshold surrogate for Willie's detector.

The numerical results support this structural picture. The robust feasible region is strictly smaller than the nominal one, the robust optimal rate decreases systematically with uncertainty, and the resulting rate loss can be substantial. In addition, the Monte Carlo results show that the midpoint-threshold and exact-threshold radiometer benchmarks are visually indistinguishable in the intended low-effective-SNR regime, while the analytical surrogate closely tracks both.

Overall, the paper provides a simple analytical framework that makes the robust-design geometry transparent and highlights the throughput penalty caused by bounded uncertainty. Although the model is deliberately streamlined, it offers a clear starting point for richer extensions, including finite-block detection analysis, exact finite-$N$ detector optimization, and broader uncertainty models.

\bibliographystyle{IEEEtran}
\bibliography{references}
\end{document}